\documentclass[a4paper]{llncs}

\usepackage{listings}
\usepackage{graphicx}
\usepackage{float,subfig}
\usepackage{xspace,framed}
\usepackage{pgf}
\usepackage[noend]{algpseudocode}
\usepackage{algorithm}
\usepackage{amsfonts}
\usepackage{amsmath}
\usepackage{pifont}
\usepackage{multicol}
\usepackage{multirow}
\usepackage{tikz}
\usetikzlibrary{arrows, automata, shapes}

\algrenewcommand\algorithmicrequire{\textbf{Precondition:}}
\algrenewcommand\algorithmicensure{\textbf{Postcondition:}}

\newcommand{\newC}{C$^-$\xspace}
\newcommand*\Let[2]{\State #1 $\gets$ #2}

\newcommand{\kalashnikov}{{\sc Kalashnikov}\xspace}

\title{Second-Order Propositional Satisfiability}
\author{Cristina David \and Daniel Kroening \and Matt Lewis}
\institute{University of Oxford}

\newenvironment{keywords}{
       \list{}{\advance\topsep by0.35cm\relax\small
       \leftmargin=0cm
       \labelwidth=0.35cm
       \listparindent=0.35cm
       \itemindent\listparindent
       \rightmargin\leftmargin}\item[\hskip\labelsep
                                     \bfseries Keywords:]}
     {\endlist}

\makeatletter
\pgfdeclareshape{datastore}{
  \inheritsavedanchors[from=rectangle]
  \inheritanchorborder[from=rectangle]
  \inheritanchor[from=rectangle]{center}
  \inheritanchor[from=rectangle]{base}
  \inheritanchor[from=rectangle]{north}
  \inheritanchor[from=rectangle]{north east}
  \inheritanchor[from=rectangle]{east}
  \inheritanchor[from=rectangle]{south east}
  \inheritanchor[from=rectangle]{south}
  \inheritanchor[from=rectangle]{south west}
  \inheritanchor[from=rectangle]{west}
  \inheritanchor[from=rectangle]{north west}
  \backgroundpath{
    \southwest \pgf@xa=\pgf@x \pgf@ya=\pgf@y
    \northeast \pgf@xb=\pgf@x \pgf@yb=\pgf@y
    \pgfpathmoveto{\pgfpoint{\pgf@xa}{\pgf@ya}}
    \pgfpathlineto{\pgfpoint{\pgf@xb}{\pgf@ya}}
    \pgfpathmoveto{\pgfpoint{\pgf@xa}{\pgf@yb}}
\usepackage{tikz}
\usetikzlibrary{arrows, automata, shapes}

    \pgfpathlineto{\pgfpoint{\pgf@xb}{\pgf@yb}}
 }
}
\makeatother

\begin{document}
\maketitle
\pagestyle{headings}  % switches on printing of running heads

\begin{abstract}
Fundamentally, every static program analyser searches for a proof through a combination of heuristics providing candidate solutions and a candidate validation technique. Essentially, the heuristic reduces a second-order problem to a first-order/propositional one, while the validation is often just a call to a SAT/SMT solver. This results in a monolithic design of such analyses that conflates the formulation of the problem with the solving process. Consequently, any change to the latter causes changes to the whole analysis. This design is dictated by the state of the art in solver technology. While SAT/SMT solvers have experienced tremendous progress, there are barely any second-order solvers. This paper takes a step towards addressing this situation by proposing a decidable fragment of second-order logic that is still expressive enough to capture numerous program analysis problems (e.g. safety proving, bug finding, termination and non-termination proving, superoptimisation). We refer to the satisfiability problem for this fragment as Second-Order SAT and show it is NEXPTIME-complete. Finally, we build a decision procedure for Second-Order SAT based on program synthesis and present experimental evidence that our approach is tractable for program analysis problems.
\end{abstract}


\section{Introduction}
Fundamentally, every static program analysis is searching for a {\em program proof}.
For safety analysers this proof takes the form of 
a program invariant \cite{DBLP:conf/popl/CousotC77}, for bug finders it's a counter-model \cite{DBLP:conf/dac/ClarkeKY03}, for termination analysis it 
can be a ranking function \cite{Floyd67}, whereas for non-termination it's a recurrence set \cite{DBLP:conf/popl/GuptaHMRX08}.
Finding each of these proofs was subject to extensive research resulting in a multitude of techniques.
%Fixed-point computation, 

The process of searching for a proof can be roughly seen as a refinement loop with 
%Most static analysers search for such a proof also share a common design.  
%This paper is motivated by the observation that most of these techniques share a common architecture, 
two phases. One phase is heuristic in nature, e.g. 
adjusting the unwinding depth (for bounded model checking \cite{DBLP:conf/dac/ClarkeKY03}),
refining the set of predicates (for predicate abstraction and interpolation \cite{DBLP:conf/cav/ClarkeGJLV00,DBLP:conf/cav/McMillan06}),
selecting a template (for template-based analyses \cite{dblp:conf/tacas/leikeh14}),
applying a widening operator (for abstract interpretation based techniques \cite{DBLP:conf/popl/CousotC77}),
whereas the other phase usually involves a call to a decision procedure. %is exact (deterministic?/precise?) and generally more computation intensive (very often it includes calls to a decision procedure).
From the perspective of the proof,
the heuristic constrains the universe of potential proofs to just one candidate 
(by fixing the unwinding bound, the template, the set of predicates, etc), which 
is then validated by the other phase.
The unknowns in the first phase are proofs (second-order entities), 
whereas the unknowns in the second phase are program variables.
Essentially, the first phase reduces a second-order problem to a first-order/propositional one.

Such a design makes it difficult to separate the problem's formulation (the second-order problem)
from parts of the solving process, resulting in 
analyses that are cluttered and fragile. %and error-prone. 
Any change to the search process causes changes to the whole analysis.
Ideally, we would like a modular design, where the search for a solution is encapsulated and, thus, completely separated from the formulation of the problem.

The existing design is dictated by the state of the art in solver technology.
While the SAT/SMT technologies nowdays allow solving industrial sized instances, 
there is hardly any progress made for second-order solvers.

% encapsulating the concept of refinement loop

\paragraph{Aim of current paper.}
In this paper, we take a step towards addressing this situation by presenting a fragment of second-order logic 
that is decidable, while still expressive enough to capture static analysis problems, such as safety proving, bug finding, termination and
non-termination proving, superoptimisation, etc. 
Decidability is regained by allowing only existential second-order quantification 
and interpreting ground terms over a finite universe.  
We call the resulting problem Second-Order SAT %(formally defined in section~\ref{sec:kalashnikov-prelim}) 
and show it is NEXPTIME-complete. 
%The resulting fragment is still expressive enough to capture program analysis problems as shown in section~\ref{}.
%After showing how many problems, including all those listed above, can be concisely and naturally encoded as second-order SAT, 
%By further exploiting the correspondence between program synthesis and second-order logic, 
%We show that Second-Order SAT is NEXPTIME-complete

By using finite state program synthesis, we build a solver for Second-Order SAT. 
A very important property of our solver is that its runtime is heavily influenced by the {\em length of the shortest
proof}, i.e. the Kolmogorov complexity of the problem as we will discuss in section~\ref{sec:kalashnikov-proofs}.
If a short proof exists, then the solver will find it quickly. 
This is particularly useful for program analysis problems,
where, if a proof exists, then most of the time many proofs exist and some are short 
(\cite{DBLP:conf/aplas/KongJDWY10} relies on a similar remark about loop invariants).
Our experiments show that for many problems the shortest solution is often very short,
which means that we avoid the NEXPTIME bound, so our approach is tractable far more often than one might reasonably expect.
%These observations are corroborated by our experimental results, as will be seen in Section~\ref{sec:kalashnikov-experiments}.
%% We will now show that the number of iterations of the CEGIS loop
%% is a function of the Kolmogorov complexity of the synthesised program.

In the same way that SAT solvers enabled a leap forwards by presenting a unified interface to an NP-complete problem,
Second-Order SAT presents a unified interface for NEXPTIME problems.  For program analysers, this means that the business of finding a solution (running a refinement loop) 
can be offloaded to a black box solver.
We advocate for using second-order logic for formulating static analysis problems, and hope that
this work will encourage researchers in this direction.
Analyses formulated this way will automatically benefit from future advances in second-order solving.

\paragraph{Related logics.}
Other second-order solvers are introduced in \cite{DBLP:conf/pldi/GrebenshchikovLPR12,DBLP:conf/cav/BeyenePR13}. As opposed to Second-Order SAT,
these are specialised for Horn clauses and the logic they handle is undecidable.
Wintersteiger et al. present in \cite{DBLP:conf/fmcad/WintersteigerHM10} a decision procedure for  
a logic related to Second-Order SAT, the Quantified bit-vector logic, which is  
a many sorted first-order logic formula where the sort of every variable is a bit-vector sort. 

\paragraph{Technical Contributions:}
\begin{itemize}
% \item We show that synthesising finite state programs is NEXPTIME-complete.
 \item We define the Second-Order SAT problem and show that it is NEXPTIME-complete.
 \item We show that Second-Order SAT is expressive enough to encode several program analysis problems.
 \item We build a decision procedure for Second-Order SAT via a reduction to finite state program synthesis. 
The decision procedure uses a novel combination of symbolic model checking, explicit state model checking and
stochastic search, and its runtime is dependent on the length of the shortest proof.
 %\item We describe a fully automatic, sound and complete algorithm for synthesising finite state programs.
 \item We evaluated our solver on several static analysis problems.
 %% \item Our implementation and all of our experimental results are freely available to download.
  %The implementation is robust and has already been used as the backend for other tools~\cite{synth-termination}.
\end{itemize}

\iffalse
Our approach generates an $\exists \forall$ formula that is linear
in the size of the synthesised program.  This improves over the
encoding of~\cite{brahma}, which generates a $\exists \forall$ formula
that is quadratic in the size of the synthesised program.

We introduce a novel parametrisation of the programming language
used to express our synthesised programs.
This parametrisation allows us to efficiently explore the program space
without relying on human guidance and also ensures that our programs
are of minimal length.

Our tool is the first we are aware of that is able to effectively
synthesise floating-point programs.  We demonstrate this by
synthesising {\sc Fast2Sum} using Knuth's {\sc 2Sum}~\cite{taocp2} as
a specification.
\fi

\iffalse
\paragraph{Outline} We describe the abstract synthesis algorithm in
Sec.~\ref{sec:kalashnikov-abstract-synth} and give details of our new algorithm in
Sec.~\ref{sec:kalashnikov-concrete-algorithm}.  We use a fragment of the C programming language
as our underlying logic, which is elaborated in Sec.~\ref{sec:kalashnikov-logic}. We
summarise the results of our experiments in Sec.~\ref{sec:kalashnikov-experiments}.
\fi

\section{Preliminaries}
\label{sec:kalashnikov-prelim}

In this section we will recall some well known decision problems along with their
associated complexity classes.  
\begin{definition}[Propositional SAT]
\label{def:sat}
 \[
  \exists x_1 \ldots x_n . \sigma
 \]
 Where the $x_i$ range over Boolean values and $\sigma$ is a quantifier-free propositional formula
 whose free variables are the $x_i$.
\end{definition}

Checking the truth of an instance of Definition~\ref{def:sat} is NP-complete.

\begin{definition}[First-Order Propositional SAT or QBF]
\label{def:qbf}
 \[
  Q_1 x_1 . Q_2 x_2 \ldots Q_n x_n . \sigma
 \]
 Where the $Q_i$ are either $\exists$ or $\forall$.  The $x_i$ and $\sigma$ are as in
 Definition~\ref{def:sat}.
\end{definition}

Checking the truth of an instance of Definition~\ref{def:qbf} is PSPACE-complete.

\section{Second-Order SAT}
%\todo{add more context}

In this section, we introduce Second-Order SAT, an extension of propositional SAT. Subsequently, we prove 
that Second-Order SAT is NEXPTIME-complete.

%Now we turn our attention to second-order logic.  Second-order logic allows quantification
%over sets as well as objects.  

\begin{definition}[Second-Order SAT]
\label{def:2sat}
 \[
  \exists S_1 \ldots S_m . Q_1 x_1 \ldots Q_n x_n . \sigma
 \]
 Where the $S_i$ range over predicates.  Each $S_i$ has an associated arity $\mathrm{ar}(S_i)$
 and $S_i \subseteq \mathbb{B}^{\mathrm{ar}(S_i)}$.  The remainder of the formula
 is an instance of Definition~\ref{def:qbf}, except that the quantifier-free part ($\sigma$)
 may refer to both the first-order variables $x_i$ and the second-order variables $S_i$.
\end{definition}

\begin{example}
 The following is a Second-Order SAT formula:

 \[
  \exists S . \forall x_1, x_2 . S(x_1, x_2) \rightarrow S(x_2, x_1)
 \]

This formula is satisfiable and is satisfied by any symmetric relation.
\end{example}

\begin{theorem}[Fagin's Theorem~\cite{fagin}]
\label{thm:fagin}
 The class of structures $A$ recognisable in time $|A|^k$, for some $k$, by a nondeterministic Turing machine
 is exactly the class of structures definable by existential second-order sentences.
\end{theorem}

\begin{theorem}[Second-Order SAT is NEXPTIME-complete]
 For an instance of Definition~\ref{def:2sat} with $n$ first-order variables,
 checking the truth of the formula is NEXPTIME-complete.
\end{theorem}

\begin{proof}
We will apply Theorem~\ref{thm:fagin}.  To do so we must establish the size of
the universe implied by Theorem~\ref{thm:fagin}.  Since Definition~\ref{def:2sat} uses
$n$ Boolean variables, the universe is the set of interpretations
of $n$ Boolean variables.  This set has size $2^n$, and so by Theorem~\ref{thm:fagin},
Definition~\ref{def:2sat} defines exactly the class sets recognisable in $(2^n)^k$
time by a nondeterministic Turing machine.  This is the class NEXPTIME, and
so checking validity of an arbitrary instance of Definition~\ref{def:2sat}
is NEXPTIME-complete.
\end{proof}

For an alternative proof, consider a Turing machine $M$.  For a particular run
of $M$ we can construct a relation $f(k, q, h, j, t)$ defined such that
after $k$ steps $M$ is in state $q$, with its head at position $h$ and
tape cell $j$ containing the symbol $t$.  If $M$ halts within $2^n$ steps
on an input of length $n$, the values of all the variables in this relation
are bounded by $2^n$, which means they can be written down using $n$ bits.
The details of creating a first-order formula constraining $f$ to reflect
the behaviour of $M$ are left to the reader.

\section{Decision Procedure for Second-Order SAT}
\label{sec:complexity}

A solution to an instance of Definition~\ref{def:2sat} is an assignment mapping
each of the second-order variables to some function of the appropriate type and arity.
When deciding whether a particular Second-Order SAT instance is satisfiable, we should think
about how solutions are encoded and in particular how a function is to be encoded.
The functions all have a finite domain and co-domain, so their canonical representation
would be a finite set of ordered pairs.  Such a set is exponentially large in the size of
the domain, so we would prefer to work with a more compact representation if possible.

We will generate \emph{finite state programs} that compute the functions and represent these
programs as finite lists of instructions in SSA form.  This representation has the following
properties, proofs for which can be found in Appendix~\ref{app:proofs}.

\begin{theorem}
\label{thm:l-universal}
Every total, finite function is computed by at least one program in this language.
\end{theorem}

\begin{theorem}
\label{thm:l-concise}
Furthermore, this representation is optimally
concise -- there is no encoding that gives a shorter representation to every function.
\end{theorem}

\paragraph{Finite State Program Synthesis}

To formally define the finite state synthesis problem, we need to fix some notation.
We will say that a program $P$ is a finite list of instructions in SSA form, where no
instruction can cause a back jump, i.e. our programs are loop free and non-recursive.
Inputs $x$ to the program are drawn from some finite domain $\mathcal{D}$.
The synthesis problem is given to us in the form of a specification $\sigma$ which is
a function taking a program $P$ and input $x$ as parameters and returning a boolean
telling us whether $P$ did ``the right thing'' on input $x$.
Basically, the finite state synthesis
problem checks the truth of Definition~\ref{def:finite-synth-formula}.

\begin{definition}[Finite Synthesis Formula]
 \label{def:finite-synth-formula}
\[
 \exists P . \forall x \in \mathcal{D} . \sigma(P, x)
\]
\end{definition}

To express the specification $\sigma$, we introduce a function
$\mathtt{exec}(P, x)$ that returns the result of running program $P$ with input $x$.
Since $P$ cannot contain loops or recursion, $\mathtt{exec}$ is a total function.

\begin{example}
The following finite state synthesis problem is satisfiable:
\[
 \exists P . \forall x \in \mathbb{N}_8 . \mathtt{exec}(P, x) \geq x
\]

One such program $P$ satisfying the specification is \verb|return 8|, which just returns 8 for any input.
\end{example}

We now present our main theorem, which says that Second-Order SAT can be reduced to finite state
program synthesis.  The proof of this theorem can be found in
Appendix~\ref{app:proofs}.

\begin{theorem}[Second-Order SAT is Polynomial Time Reducible to Finite Synthesis]
 Every instance of Definition~\ref{def:2sat} is polynomial time reducible to an instance
 of Definition~\ref{def:finite-synth-formula}.
\end{theorem}

\begin{corollary}
\label{cor:finite-synth-nexp}
 Finite-state program synthesis is NEXPTIME-complete.
\end{corollary}

We are now in a position to sketch the design of a decision procedure for Second-Order SAT: we will convert
the Second-Order SAT problem to an equisatisfiable finite synthesis problem which we will then solve
with a finite state program synthesiser.  This design will be elaborated in Section~\ref{sec:kalashnikov-implementation}.

\section{Solving Second-Order SAT via Finite-State Program Synthesis}
\label{sec:kalashnikov-implementation}

%\todo{make this section much shorter}

In this section we will present a sound and complete algorithm for
finite-state synthesis that we use to decide Second-Order SAT.
%In the case that a specification
%is satisfiable, our algorithm produces a minimal satisfying
%program \todo{does this still hold with GP?}. 
We begin by describing a general purpose synthesis procedure (Section~\ref{sec:kalashnikov-abstract-synth}),
then detail how this general purpose procedure is instantiated for
synthesising finite-state programs.  %For the latter part, we will
%describe the logic on which our system is built and
%how programs are encoded in our system (Section~\ref{sec:kalashnikov-encode-l}).
We then describe the algorithm
we use to search the space of possible programs (Sections~\ref{sec:kalashnikov-candidates}, \ref{sec:kalashnikov-param-space} and~\ref{sec:kalashnikov-search-space}).
%some optimisations we found to be essential (Section~\ref{sec:kalashnikov-optimisations}).
%and conclude with a proof of soundness and
%complexity bounds (Section~\ref{sec:kalashnikov-proofs}).

\subsection{General Purpose Synthesis Algorithm}
\label{sec:kalashnikov-abstract-synth}

We use Counterexample Guided Inductive Synthesis
(CEGIS)~\cite{lezama-thesis,sketch,toast} to find a
program satisfying our specification.  The core of the CEGIS
algorithm is the refinement loop %shown in Figure~\ref{fig:abstract-refinement} and
detailed in Algorithm~\ref{alg:cegis}.

\begin{algorithm}
 \caption{Abstract refinement algorithm
 \label{alg:cegis}}

 \begin{multicols}{2}
 \begin{algorithmic}[1]
\Statex
\Function{synth}{inputs}
  \Let{$(i_1, \ldots, i_N)$}{inputs}
  \Let{query}{$\exists P . \sigma(i_1, P) \land \ldots \land \sigma(i_N, P)$}
  \Let{result}{decide(query)}
  \If{result.satisfiable}
    \State \Return{result.model}
  \Else
    \State \Return{UNSAT}
  \EndIf
\EndFunction
\Statex
\Function{verif}{P}
  \Let{query}{$\exists x . \lnot \sigma(x, P)$}
  \Let{result}{decide(query)}
  \If{result.satisfiable}
    \State \Return{result.model}
  \Else
    \State \Return{valid}
  \EndIf
\EndFunction
\columnbreak
\Statex
\Function{refinement loop}{}
  \Let{inputs}{$\emptyset$}
  \Loop
    \Let{candidate}{\Call{synth}{inputs}}
    \If{candidate = UNSAT}
      \State \Return{UNSAT}
    \EndIf
    \Let{res}{\Call{verif}{candidate}}
    \If{res = valid}
      \State \Return{candidate}
    \Else
      \Let{inputs}{inputs $\cup$ res}
    \EndIf
  \EndLoop
\EndFunction
 \end{algorithmic}
 \end{multicols}
\end{algorithm}

%% \begin{figure}
%%  \centering
%%  \begin{tikzpicture}[scale=0.5,->,>=stealth',shorten >=1pt,auto,
%%  semithick, initial text=]

%%   \matrix[nodes={draw, fill=none, scale=1, shape=rectangle, minimum height=1cm, minimum width=1.5cm},
%%           row sep=2cm, column sep=3cm] {
%%    \node (synth) {Synthesise};
%%    &
%%    \node (verif) {Verify}; %\\
%%    %\node[draw=none] {};
%%    &
%%    \node[ellipse] (done) {Done}; \\
%%   };

%%    \path
%%     (synth) edge [bend left] node {Candidate program} (verif)
%%     (verif) edge [bend left] node {Counterexample input} (synth)
%%     (verif) edge node {Valid} (done);
%%  \end{tikzpicture}
 
%%  \caption{Abstract synthesis refinement loop
%%  \label{fig:abstract-refinement}}
%% \end{figure}

The algorithm is divided into two procedures: {\sc synth} %(see Figure~\ref{fig:synth-dfd}) 
and {\sc verif}, which interact via a finite set of test vectors {\sc inputs}.
The {\sc synth} procedure tries to find an existential witness $P$
that satisfies the partial specification:
\[
 \exists P . \forall x \in \text{\sc inputs} . \sigma(x, P)
\]

If {\sc synth} succeeds in finding a witness $P$, this witness is a
candidate solution to the full synthesis formula.  We pass this candidate
solution to {\sc verif} which determines whether it does satisfy
the specification on all inputs by checking satisfiability of the
verification formula:
\[
 \exists x . \lnot \sigma(x, P)
\]

If this formula is unsatisfiable, the candidate solution is in fact a
solution to the synthesis formula and so the algorithm terminates. 
Otherwise, the witness $x$ is an input on which the candidate solution fails
to meet the specification.  This witness $x$ is added to the {\sc inputs}
set and the loop iterates again.  It is worth noting that each iteration of the
loop adds a new input to the set of inputs being used for synthesis.  If
the full set of inputs $X$ is finite, this means that the refinement loop
can only iterate a finite number of times.

\subsection{Finite-State Synthesis}
\label{sec:kalashnikov-concrete-algorithm}

\iffalse
One area in which program synthesis can shine is in producing very small,
intricate programs that manipulate bitvectors.  An example of such a program
is given in Fig.~\ref{fig:bitvector-program}.  This program takes a machine word
as input and clears every bit except for the least significant bit that was set.
Even though this program is extremely short, it is fairly difficult for a human
to see what it does.  It is even more difficult for a human to come up with such
minimal code.  The program is so concise because it
takes advantage of the low-level details of the machine, such as the fact that
signed integers are stored in two's complement form.
\fi

\iffalse
To synthesise tricky bitvector programs like this, it is natural for us to
work in the logic of quantifier-free propositional formulae and to use a
propositional SAT or SMT-$\mathcal{BV}$ solver as the decision procedure. 
However, we propose a slightly different tack, which is to use a decidable
fragment of C as a ``high level'' logic.
\fi

We will now show how the generic construction of Section~\ref{sec:kalashnikov-abstract-synth}
can be instantiated to produce a useful finite-state program synthesiser.
A natural choice for such a synthesiser would be to
work in the logic of quantifier-free propositional formulae and to use a
propositional SAT or SMT-$\mathcal{BV}$ solver as the decision procedure. 
However we propose a slightly different tack, which is to use a decidable
fragment of C as a ``high level'' logic.  We call this fragment \newC.

%\subsubsection{\newC}
%\label{sec:kalashnikov-logic}

%We will now describe the logic we use to express our synthesis formula.
%The logic is a subset of C that we call \newC.  
The characteristic property of a
\newC  program is that safety can be decided for
it using a single query to a Bounded Model Checker.  A \newC program is
just a C program with the following syntactic restrictions:
\begin{itemize}
 \item all loops in the program must have a constant bound;
 \item all recursion in the program must be limited to a constant depth;
 \item all arrays must be statically allocated (i.e. not using \texttt{malloc}),
 and be of constant size.
\end{itemize}
\newC programs may use nondeterministic values, assumptions
and arbitrary-width types.

\iffalse
Two example \newC programs are shown 
in Fig.~\ref{fig:c-}.

\begin{figure}
\begin{minipage}[scale=0.8]{0.45\linewidth}
 \begin{lstlisting}[language=c]
int count_bits(int x) {
  int i, ret = 0;
  
  for (i = 0; i < 32; i++)
    if (x & (1 << i))
      ret++;
  
  return ret;
}
 \end{lstlisting}
\end{minipage}
\begin{minipage}{0.54\linewidth}
 \begin{lstlisting}[language=C]
int common_factor(int A[10]) {
  int i, factor = nondet();

  for (i = 0; i < 10; i++)
    assume((A[i] % factor) == 0);

  assume(factor > 1);
  return factor;
}

 \end{lstlisting}
\end{minipage}

 \caption{Two \newC programs}
 \label{fig:c-}

\end{figure}
\fi

Since each loop is bounded by a constant, and each recursive function call is
limited to a constant depth, a \newC program necessarily terminates and in
fact does so in $O(1)$ time.  If we call the largest loop bound~$k$, then
a Bounded Model Checker with an unrolling bound of $k$ will be a complete
decision procedure for the safety of the program.  For a \newC program of
size $l$ and with largest loop bound~$k$, a Bounded Model Checker will
create a SAT problem of size $O(lk)$.  Conversely, a SAT problem
of size $s$ can be converted trivially into a loop-free \newC program
of size $O(s)$.  The safety problem for \newC is therefore NP-complete,
which means it can be decided fairly efficiently for many practical
instances.

\subsection{Encoding the Problem in \newC}
To instantiate the abstract synthesis algorithm in \newC, we must express $X, Y,
\sigma$ in \newC, then express the validity of the
synthesis formula as a safety property of the resulting \newC program.

Our encoding for these pieces is the following:
\begin{itemize}
 \item $X$ is the set of $N$-tuples of 32-bit bitvectors.  %This is written in \newC as the type \verb|int[N]|.
 \item $Y$ is the set of $M$-tuples of 32-bit bitvectors.%, which is written in \newC as the type \verb|int[M]|.
 \item $\sigma$ is a pure function with type $X \times Y \rightarrow \mathrm{Bool}$.  
%The \newC signature of this function is
% \verb|int check(int in[N], int out[M])|. This function is the only component supplied
% by the user.
\end{itemize}

 $P$ is written in a simple RISC-like language $\mathcal{L}$, whose syntax is given in Fig.~\ref{fig:l-language}.  Programs in $\mathcal{L}$
 have type $X \rightarrow Y$. %and
 %are represented in \newC as objects of type \verb|prog_t|, shown in Fig.~\ref{fig:c-l-encoding}.
 We supply an interpreter for $\mathcal{L}$ which is written in \newC.  The type of
 this interpreter is $(X \rightarrow Y) \times X \rightarrow Y$.% and the \newC signature is \\
% \verb|void exec(prog_t p, int in[N], int out[M])|.  Here, \verb|out| is an output parameter.
The specification function $\sigma$ will include calls to this interpreter, by which means it
will examine the behaviour of a candidate $\mathcal{L}$ program.

\begin{figure}
\begin{center}
{\small

\setlength{\tabcolsep}{14pt}
Integer arithmetic instructions:

\begin{tabular}{llll}
 \verb|add a b| & \verb|sub a b| & \verb|mul a b| & \verb|div a b| \\
 \verb|neg a| &   \verb|mod a b| & \verb|min a b| & \verb|max a b|
\end{tabular}

\medskip

Bitwise logical and shift instructions:

\begin{tabular}{lll}
 \verb|and  a b| & \verb|or   a b| & \verb|xor a b| \\
 \verb|lshr a b| & \verb|ashr a b| & \verb|not a|
\end{tabular}

\medskip

Unsigned and signed comparison instructions:

\begin{tabular}{lll}
 \verb|le  a b| & \verb|lt  a b| & \verb|sle  a b| \\
 \verb|slt a b| & \verb|eq  a b| & \verb|neq  a b| \\
\end{tabular}

\medskip

Miscellaneous logical instructions:

\begin{tabular}{lll}
 \verb|implies a b| & \verb|ite a b c| &  \\
\end{tabular}

\medskip

\setlength{\tabcolsep}{12pt}

Floating-point arithmetic:

\begin{tabular}{llll}
 \verb|fadd a b| & \verb|fsub a b| & \verb|fmul a b| & \verb|fdiv a b|
\end{tabular}

}
\end{center}

 \caption{The language $\mathcal{L}$}
 \label{fig:l-language}
\end{figure}

%% The exact details of how we encode an $\mathcal{L}$-program are given in Sec.~\ref{sec:kalashnikov-encode-l}.
%% We must now express the {\sc synth} and {\sc verif} formulae as safety properties
%% of \newC programs, which is given in Fig.~\ref{fig:c-synthverif}.

%% \begin{figure}
%% \centering
%% \begin{minipage}[t]{.45\textwidth}
%% \begin{lstlisting}[language=C++]
%% void synth() {
%%   prog_t p = nondet();
%%   int in[N], out[M];

%%   assume(wellformed(p));

%%   in = test1;
%%   exec(p, in, out);
%%   assume(check(in, out));
%%   ...
%%   in = testN;
%%   exec(p, in, out);
%%   assume(check(in, out));
  
%%   assert(false);
%% }
%% \end{lstlisting}
%% \end{minipage}
%% \hfill
%% \begin{minipage}[t]{.45\linewidth}
%% \begin{lstlisting}[language=C++]
%% void verif(prog_t p) {
%%   int in[N] = nondet();
%%   int out[M];

%%   exec(p, in, out);
%%   assert(check(in, out));
%% }
%% \end{lstlisting}
%% \end{minipage}

%%  \caption{The {\sc synth} and {\sc verif} formulae expressed as a \newC program.}
%%  \label{fig:c-synthverif}
%% \end{figure}

With these pieces in place, we construct a \newC program {\sc synth.c}
which takes as parameters a candidate program $P$ and test
inputs $X$.  The program contains an assertion which fails
iff $P$ meets the specification for each of the inputs in $X$.
Finding a new candidate program is then equivalent to checking
the safety of {\sc synth.c} for which we use the strategies described in the next section.%a Bounded Model

\subsection{Candidate Generation Strategies}
\label{sec:kalashnikov-candidates}
The remit of the {\sc synth} portion of the CEGIS loop %, as shown in Figure~\ref{fig:abstract-refinement}, 
is to generate candidate
programs.  There are many possible strategies for finding these candidates; we employ the
following strategies in parallel:
\paragraph{Explicit Proof Search.} The simplest strategy for finding candidates
is to just exhaustively enumerate them all, starting with the shortest and
progressively increasing the number of instructions.  %This strategy
%is implemented by the {\sc ExplicitSearch} routine.  
Since the set of
$\mathcal{L}$-programs is recursively enumerable, this procedure is complete.
\paragraph{Symbolic Bounded Model Checking.} Another complete method for generating
candidates is to simply use BMC on the {\sc synth.c} program.  As with explicit
search, we must progressively increase the length of the $\mathcal{L}$-program we search for
in order to get a complete search procedure.
\paragraph{Genetic Programming and Incremental Evolution.} \label{sec:gp}
Our final strategy is genetic programming
(GP)~\cite{langdon:fogp,brameier2007linear}.  GP provides an adaptive way of
searching through the space of $\mathcal{L}$-programs for an individual
that is ``fit'' in some sense.  We measure the
fitness of an individual by counting the number of tests in {\sc inputs}
for which it satisfies the specification.

To bootstrap GP in the first iteration of the CEGIS loop, we generate a population
of random $\mathcal{L}$-programs. We then iteratively evolve this population by
applying the genetic operators {\sc crossover} and {\sc mutate}.
{\sc Crossover} combines selected existing programs into new programs,
whereas {\sc mutate} randomly changes parts of a single program.
Fitter programs are more likely to be selected.

% The {\sc crossover} operator takes two programs and combines their
% code in some way to create a new program.  Programs are selected to be
% {\sc crossover}ed according to their fitness -- fitter programs are more
% likely to be selected to become parents.

% The {\sc mutate} operator takes a single program and randomly changes parts
% of its code, again creating a new program.

Rather than generating a random population at the beginning of each subsequent
iteration of the CEGIS loop, we start with the population we had at the end of the
previous iteration.  The intuition here is that this population contained
many individuals that performed well on the $k$ inputs we had before, so
they will probably continue to perform well on the $k+1$ inputs we have now.
In the parlance of evolutionary programming, this is known as
incremental evolution~\cite{Gomez97incrementalevolution}.

\subsection{Parameterising the Program Space}
\label{sec:kalashnikov-param-space}
In order to search the space of candidate programs, we parametrise
the language~$\mathcal{L}$, inducing a lattice of progressively
more expressive languages.  We start by attempting to synthesise
a program at the lowest point on this lattice and increase the
parameters of~$\mathcal{L}$ until we reach a point at which
the synthesis succeeds.

As well as giving us an automatic search procedure, this parametrisation
greatly increases the efficiency of our system since languages
low down the lattice are very easy to decide safety for.  If a program
can be synthesised in a low-complexity language, the whole procedure
finishes much faster than if synthesis had been attempted in a
high-complexity language.

\paragraph{Program Length: $l$}
The first parameter we introduce is program length, denoted by $l$.
At each iteration we synthesise programs of length exactly $l$.
We start with $l = 1$ and increment $l$ whenever we determine
that no program of length $l$ can satisfy the specification.  When we do
successfully synthesise a program, we are \emph{guaranteed that it
is of minimal length} since we have previously established that no
shorter program is correct.

\paragraph{Word Width: $w$}
An $\mathcal{L}$-program runs on a virtual machine (the $\mathcal{L}$-machine) that
has its own set of parameters.  The only relevant parameter is
the \emph{word width} of the $\mathcal{L}$-machine, that is, the number of bits
in each internal register and immediate constant.  This parameter is denoted by
$w$.  The size of the final SAT problem generated by {\sc cbmc} scales
polynomially with $w$, since each intermediate C variable corresponds
to $w$ propositional variables.

\paragraph{Number of Constants: $c$}
Instructions in $\mathcal{L}$ take either one or two operands.
Since any instruction whose operands are all constants can always be
eliminated (since its result is a constant), we know that a loop-free program
of minimal length will not contain any instructions with two constant
operands.  Therefore the number of constants that can appear in
a minimal program of length $l$ is at most $l$.  By minimising the number
of constants appearing in a program, we are able to use a particularly
efficient program encoding that speeds up the synthesis procedure
substantially.  The number of constants used in a program is the parameter $c$.

\subsection{Searching the Program Space}
\label{sec:kalashnikov-search-space}

The key to our automation approach is to come up with a sensible way in which to
adjust the $\mathcal{L}$-parameters in order to cover all possible programs.
After each round of {\sc synth}, we may need to adjust the parameters.  The
logic for these adjustments is shown as a tree in Fig.~\ref{fig:paramsflow}.

Whenever {\sc synth} fails, we consider which parameter might have caused the
failure.  There are two possibilities: either the program length $l$ was too small,
or the number of allowed constants $c$ was.  If $c < l$, we just increment $c$ and
try another round of synthesis, but allowing ourselves an extra program constant.
If $c = l$, there is no point in increasing $c$ any further.  This is because
no minimal $\mathcal{L}$-program has $c > l$, for if it did there would
have to be at least one instruction with two constant operands.  This
instruction could be removed (at the expense of adding its result as
a constant), contradicting the assumed minimality of the program.  So
if $c = l$, we set $c$ to 0 and increment $l$, before attempting
synthesis again.

If {\sc synth} succeeds but {\sc verif} fails, we have a candidate
program that is correct for some inputs but incorrect on at least
one input.  However, it may be the case that the candidate program
is correct for \emph{all} inputs when run on an $\mathcal{L}$-machine
with a small word size.  For example, we may have synthesised a
program which is correct for all 8-bit inputs, but incorrect for
some 32-bit input.  If this is the case (which we can determine
by running the candidate program through {\sc verif} using the smaller
word size), we may be able to produce a correct program for
the full $\mathcal{L}$-machine by using some constant extension rules.
%shown in Fig.~\ref{fig:generalize}.  
If constant generalization
is able to find a correct program, we are done.  Otherwise,
we need to increase the word width of the $\mathcal{L}$-machine
we are currently synthesising for.

\begin{figure}[t]
\centering
\begin{tikzpicture}[scale=0.7, transform shape, node distance=3cm, auto]
\tikzstyle{decision} = [diamond, draw, fill=blue!20, 
    text width=4.5em, text badly centered, node distance=3cm, inner sep=0pt]
\tikzstyle{block} = [rectangle, draw, fill=blue!20, 
    text width=5em, text centered, rounded corners, minimum height=4em]
\tikzstyle{line} = [draw, -latex']
\tikzstyle{cloud} = [draw, ellipse,fill=red!20, node distance=3cm,
    minimum height=2em]

 \node [decision] (synthsucc) {{\sc Synth} succeeds?};

 \node [decision, below of=synthsucc, node distance=3cm] (verif) {{\sc Verif} succeeds?};
 \node [decision, right of=verif] (ck) {$c < l$?};

 \node [block, below of=verif, node distance=4cm]   (done) {Done!};
 \node [decision, left of=verif, node distance=4cm] (verifw) {{\sc Verif} succeeds for small words?};

 \node [block, below of=ck, node distance=4cm] (incc) {$c := c+1$};
 \node [block, right of=incc, node distance=3cm] (incl) {$c := 0$\\ $l := l+1$};

 \node [decision, below of=verifw, node distance=4cm] (gen) {Extend succeeds?};
 \node [block, left of=verifw, node distance=4cm] (iterate) {Parameters unchanged};

 \node [block, left of=gen, node distance=4cm] (incw) {$w := w+1$};

 \path [line] (synthsucc) -- node [left] {Yes} (verif);
 \path [line] (synthsucc) -| node [above, near start] {No} (ck);

 \path [line] (verif) -- node [left] {Yes} (done);
 \path [line] (verif) -- node [above, near start] {No} (verifw);

 \path [line] (ck) -- node [left] {Yes} (incc);
 \path [line] (ck) -| node  [above, near start]  {No} (incl);

 \path [line] (verifw) -- node [right] {Yes} (gen);
 \path [line] (verifw) -- node [above] {No} (iterate);
 
 \path [line] (gen) -- node [below] {Yes} (done);
 \path [line] (gen) -- node [below] {No} (incw);

 %\path [dotted, line] (iterate.west) |- (synthsucc);
\end{tikzpicture}

 \caption{Decision tree for increasing parameters of $\mathcal{L}$.}
 \label{fig:paramsflow}

\end{figure}
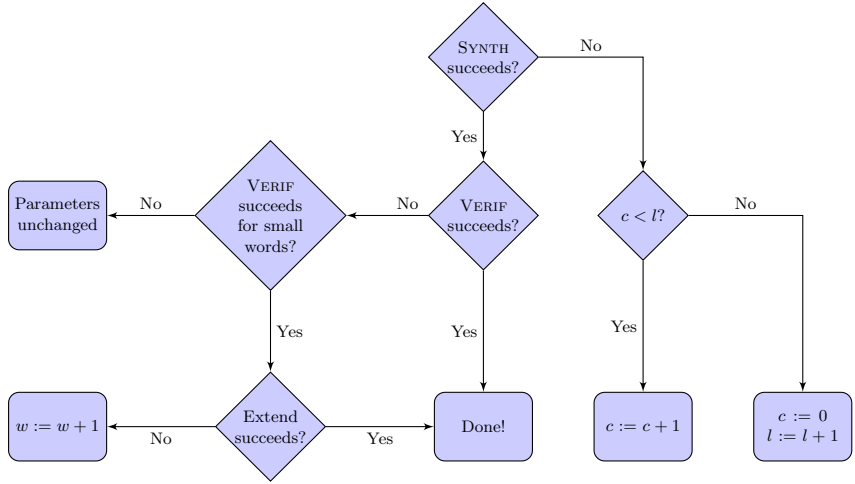

\subsection{Stopping Condition for Unsatisfiable Specifications}
\label{sec:stopping-condition}
If a specification is unsatisfiable, we would still like our algorithm to terminate
with an ``unsatisfiable'' verdict.  To do this, we can observe that any total function
taking $n$ bits of input is computed by some program of at most $2^n$ instructions
(a consequence of Theorems~\ref{thm:l-universal} and~\ref{thm:l-concise}).
Therefore every satisfiable specification has a solution with at most $2^n$ instructions.
This means that if we ever need to increase the length of the candidate program we
search for beyond $2^n$, we can terminate, safe in the knowledge that the
specification is unsatisfiable.

%% \subsection{Optimisations}
%% \label{sec:kalashnikov-optimisations}
%% Two optimisations we have found to be very important to
%% the performance of our synthesiser are the following:

%% \paragraph{Cache binaries} We ensure that we do not run {\sc gcc}
%% more times than necessary, since we have observed compilation time
%% to be relatively expensive.  This means that for the phases using
%% native code (explicit-state model checking and the stochastic methods),
%% we compiled the specification once and then just execute the resulting
%% binary in each iteration of the {\sc synth} and {\sc verif} phases.

%% \paragraph{Emit C code when possible} In the {\sc verif} stage,
%% we can emit the struct-based representation of an $\mathcal{L}$-program
%% along with the code for the interpreter and check the resulting
%% program.  Alternatively, since an $\mathcal{L}$-program can be
%% trivially translated to a straight-line C program, we can emit the
%% program as C instead.  This results in a much smaller program
%% that is more amenable to optimisation by the compiler and {\sc cbmc}.

%\input{synth-complexity}
\section{Soundness and Completeness}
\label{sec:kalashnikov-proofs}
We will now state soundness and completeness results for the Second-Order SAT solver.
Proofs for each of these theorems can be found in Appendix~\ref{app:proofs}.
%Since we have shown a reduction from second-order SAT to finite state program synthesis
%in Section~\ref{sec:complexity}, this also constitutes a proof of soundness and completeness
%for second-order SAT.

\begin{theorem}\label{thm:synth-sound}
Algorithm~\ref{alg:cegis} is sound -- if it terminates with witness $P$, then
$P \models \sigma$.
\end{theorem}

\begin{theorem}
 \label{thm:synth-semi-complete}
 Algorithm~\ref{alg:cegis} is semi-complete -- if a solution $P \models \sigma$
 exists then Algorithm~\ref{alg:cegis} will find it.
\end{theorem}

\begin{theorem}
 \label{thm:finite-synth-complete}
 Algorithm~\ref{alg:cegis} with the stopping condition described in Section~\ref{sec:stopping-condition}
 is complete when instantiated with \newC as a background theory -- it will terminate for all specifications $\sigma$.
\end{theorem}

 Since safety of \newC programs is decidable, Algorithm~\ref{alg:cegis} is sound and complete
 when instantiated with \newC as a background theory and using the stopping condition of
 Section~\ref{sec:stopping-condition}.  This construction therefore gives as a decision procedure
 for Second-Order SAT.

\paragraph{Runtime as a Function of Solution Size.}
We note that the runtime of our solver is heavily influenced by the length of the shortest
program satisfying the specification, since we begin searching for short programs.
%In fact we can rephrase our complexity as $NP^{NP}$ in the size of the shortest solution.
%This bound is equivalent to NEXPTIME since the shortest solution may be exponentially large
%in the size of the specification.
%
%We have observed that for many real world problems the shortest solution is actually very short,
%which means that we avoid the NEXPTIME bound, so our approach is tractable far more often than one might reasonably expect.
%These observations are corroborated by our experimental results, as will be seen in Section~\ref{sec:kalashnikov-experiments}.
%
We will now show that the number of iterations of the CEGIS loop
is a function of the Kolmogorov complexity of the synthesised program.
%We argue that this gives our procedure various desirable qualities in practical
%applications.
Let us first recall the definition of the Kolmogorov complexity
of a function $f$:

\begin{definition}[Kolmogorov complexity]
 The Kolmogorov complexity $K(f)$ is the length of the shortest program that 
 computes~$f$.
\end{definition}

We can extend this definition slightly to talk about the Kolmogorov complexity of a
synthesis problem in terms of its specification:

\begin{definition}[Kolmogorov complexity of a synthesis problem]
 The Kolmogorov complexity of a program specification $K(\sigma)$ is the length of the shortest
 program $P$ such that P is a witness to the satisfiability of $\sigma$.
\end{definition}

Let us consider the number of iterations of the CEGIS loop $n$ required for a specification
$\sigma$.  Since we enumerate candidate programs in order of length, we are always synthesising
programs with length no greater than $K(\sigma)$ (since when we enumerate the first correct program,
we will terminate).  So the space of solutions we search over is the space
of functions computed by $\mathcal{L}$-programs of length no greater than $K(\sigma)$.  Let's
denote this set $\mathcal{L}(K(\sigma))$.
Since there are $O(2^{K(\sigma)})$ \emph{programs} of length $K(\sigma)$ and some functions
will be computed by more than one program, we have $| \mathcal{L}(K(\sigma)) | \leq O(2^{K(\sigma)})$.

Each iteration of the CEGIS loop distinguishes at least one incorrect function from the set of correct
functions, so the loop will iterate no more than $| \mathcal{L}(K(\sigma)) |$ times.
Therefore another bound on our runtime is
$NTIME\left(2^{K(\sigma)} \right)$.

\section{Applications of Second-Order SAT}

Program analysis problems can be reduced to the problem of finding
solutions to a second-order constraint \cite{DBLP:conf/pldi/GulwaniSV08,DBLP:conf/pldi/GrebenshchikovLPR12,synth-termination}.
In this section we will show that, although decidable, Second-Order SAT is still expressive enough to capture 
many interesting such problems. Moreover, since Second-Order SAT is NEXPTIME-complete, any NEXPTIME problem can be solved with a Second-Order SAT solver.
When we describe analyses related to loops, we will characterise the loop
as having initial state $I$, guard $G$, transition relation $B$.

\paragraph{Safety Invariants.}
Given a safety assertion $A$, a safety invariant is a set of states $S$ which is inductive with respect to the program's transition
relation, and which excludes an error state.  A predicate $S$ is a safety invariant iff it satisfies the following criteria:%
\begin{align}
  \exists S . \forall x, x' .  & I(x) \rightarrow S(x) ~ \wedge \label{safety_base}\\
  & S(x) \wedge G(x) \wedge B(x, x') \rightarrow S(x') ~ \wedge \label{safety_inductive}\\
  & S(x) \wedge \neg G(x) \rightarrow A(x) \label{safety_safe}
\end{align}

%% \begin{align}
%%  \forall x . & I(x) \rightarrow S(x) \label{safety_base} \\
%%  \forall x, x' . & S(x) \wedge G(x) \wedge B(x, x') \rightarrow S(x') \label{safety_inductive} \\
%%  \forall x . & S(x) \wedge \neg G(x) \rightarrow A(x) \label{safety_safe}
%% \end{align}

\ref{safety_base} says that each state reachable on entry to the loop is in the set $S$, and in
combination with \ref{safety_inductive} shows that every state that can be reached by the loop
is in $S$.  The final criterion~\ref{safety_safe} says that if the loop exits while in an $S$-state,
the assertion $A$ is not violated.  %The existence of a safety invariant corresponds to the notion of
%partial correctness: no assertion will fail, but the program may never stop running.  %% If we
%% allow ourselves to quantify over predicates, we can specify the existence of a safety invariant
%% with {\bf [SI]} from Definition~\ref{def:si}.
%% If we restrict the language of predicates to
%% \Llang-expressions and the relations $I, G, B, A$ to \newC programs, Definition~\ref{def:si} is a second-order
%% SAT formula that we can solve directly with Kalashnikov.

%% \begin{figure}
%% \begin{framed}
%% \begin{definition}[Safety Invariant Formula {\bf [SI]}]
%% \label{def:si}
%% \begin{align*}
%%   \exists S . \forall x, x' .  & I(x) \rightarrow S(x) ~ \wedge \\
%%   & S(x) \wedge G(x) \wedge B(x, x') \rightarrow S(x') ~ \wedge \\
%%   & S(x) \wedge \neg G(x) \rightarrow A(x)
%% \end{align*}
%% \end{definition}

%% \caption{Existence of a safety invariant as second-order SAT\label{fig:safety_formula}}
%% \end{framed}
%% \end{figure}

\paragraph{Termination and non-termination.}
As shown in~\cite{synth-termination}, termination of a loop can be encoded as the following
Second-Order SAT formula, where $W$ is an inductive invariant of the loop that is established by the initial states $I$ if the
loop guard $G$ is met, and $R$ is a ranking function as restricted by $W$: 
%
%\begin{definition}[Conditional Termination Formula {\bf [CT]}]
%\label{def:ct}
 \begin{align*}
  \exists R, W . \forall x, x' . & I(x) \wedge G(x) \rightarrow W(x) ~ \wedge \\
                                 & G(x) \wedge W(x) \wedge B(x, x') \rightarrow W(x') \wedge R(x) > 0 
  \wedge R(x) > R(x')
 \end{align*}
%\end{definition}
%
Similarly, non-termination can be expressed in Second-Order SAT as follows:
%
%\begin{definition}[Non-Termination Formula -- Skolemized Open Recurrence Set  {\bf [SNT]}]
\label{def:snt}
 \begin{align*}
  \exists N, C, x_0 . \forall x .  N(x_0) ~\wedge~  &  N(x) \rightarrow G(x) ~ \wedge \\
							& N(x) \rightarrow B(x, C(x)) \wedge N(C(x))
 \end{align*}
%\end{definition}
%
$N$ denotes a recurrence set, i.e. a nonempty set of
states such that for each $s \in N$ there exists a transition to some $s' \in N$, and
$C$ is a Skolem function that chooses the successor $x'$.
More details on the formulations for termination and non-termination can be found in \cite{synth-termination}.

%% The loop terminates iff Definition~\ref{def:ct} is satisfiable and does not terminate iff Definition~\ref{def:snt}
%% is satisfiable.  As discussed in~\cite{synth-termination}, encoding the analysis problem as complementary
%% Second-Order SAT problems enables us to find a proof much more rapidly than.

\paragraph{Other NEXPTIME Problems.}
NEXPTIME is a very large complexity class, which includes decision problems over infinite-state and finite-state systems.
Some other interesting problems in NEXPTIME, and which are therefore reducible to Second-Order SAT, include:

\begin{itemize}
 \item Satisfiability of modal-$\mu$ calculus\cite{kozen1983results}.
 \item Reachability of pushdown systems~\cite{bouajjani1997reachability}.
 \item Finding a winning strategy for Go~\cite{DBLP:conf/ifip/Robson83}.
\end{itemize}

Problems in any of the classes included in NEXPTIME are of course reducible to Second-Order SAT. 
For instance, we have evaluated our solver on the Quantified Boolean Formula (QBF) problem, a well-known PSPACE-complete problem.  

\section{Experiments}
\label{sec:kalashnikov-experiments}

We implemented our decision procedure for Second-Order SAT as %fully automatic synthesis procedure in 
the \kalashnikov tool.
%The specification language of \kalashnikov is \newC, which is rich enough to encode
%arbitrary second-order SAT formulae.  
To evaluate the viability of Second-Order SAT,
we used \kalashnikov to solve formulae generated from a variety of problems.  Our
benchmarks come from superoptimisation, code deobfuscation, floating point verification,
ranking function and recurrent set synthesis, and QBF solving.
The superoptimisation and code deobfuscation benchmarks were taken from the experiments
of~\cite{brahma}; the termination benchmarks were taken from SVCOMP'15~\cite{svcomp15} and
they include the experiments of~\cite{synth-termination}; the QBF instances consist of some simple instances
created by us and some harder instances taken from~\cite{qbflib}.

We would like to stress that these experiments serve to evaluate the potential of using
Second-Order SAT as a backend for many program analysis tasks and are not intended
to compare performance with specialised solvers for each of these tasks.

We ran our experiments on a 4-core, 3.30\,GHz Core i5 with 8\,GB of RAM.  Each benchmark
was run with a timeout of 180s.  The results are shown in Table~\ref{tbl:results}. For each category of benchmarks, we report the total
number of benchmarks in that category, the number we were able to solve within the time limit,
the average specification size (in lines of code), the average solution size (in instructions),
the average number of iterations of the CEGIS loop, the average time and total
time taken. The deobfuscation and floating point benchmarks are considered together with the superoptimisation ones.  

It should be understood that
in contrast to less expressive logics that might be invoked several times in the analysis of some
problem, each of these benchmarks is a ``complete'' problem from
the given problem domain.  For example, each of the benchmarks in the termination category
requires \kalashnikov to prove that a full program terminates, i.e.~it must find a ranking function
and supporting invariants, then prove that these constitute a valid termination proof for the program
being analysed.

\paragraph{Discussion of the experimental results.}
The timings show that for the instances where we can find a satisfying assignment,
we tend to do so quite quickly (on the order of a few seconds).  Furthermore the
programs we synthesise are often short, even when the problem domain is very complex, such as for
termination or QBF.

Not all of these benchmarks are satisfiable, and in particular around half of the termination
benchmarks correspond to attempted proofs that non-terminating programs terminate and vice versa.
This illustrates one of the current shortcomings of Second-Order SAT as a decision procedure:
we can only conclude that a formula is unsatisfiable once we have examined candidate solutions
up to a very high length bound.  Being able to detect unsatisfiability of a Second-Order SAT 
formula earlier than this would be extremely valuable.  We note that for some formulae we can
simultaneously search for a proof of satisfiability and of unsatisfiability.  For example, since QBF is
closed under negation, we can take a QBF formula $\phi$ then encode both $\phi$ and $\neg \phi$ as
second-order SAT formulae which we then solve.

\begin{table}[h]
\resizebox{\textwidth}{!}{
 \begin{tabular}{|l||l|c|c|c|c|c|c|}
 \hline
   Category & \#Benchmarks & \#Solved & Spec. size & Solution size & Iterations & Avg. time (s) & Total time (s) \\
   \hline
   \hline
   Superoptimisation & 29 & 22 & 19.0 & 4.1 & 2.7 & 7.9 & 166.1 \\
   \hline
   Termination & 78 & 33 & 93.5 & 5.7 & 14.4 & 11.8 & 390.4 \\
   \hline
   QBF (simple) & 4 & 4 & 12.2 & 9 & 1.0 & 1.8 & 7.1 \\
   \hline
   QBF (hard) & 7 & 1 & 5889.0 & 11.0 & 2.0 & 1.5 & 1.5 \\
   \hline
   \hline
   Total & 113 & 59 & 49116 & 295 & 536 & --- & 565.2 \\
   \hline
 \end{tabular}
}
\vspace{1em}

 \caption{Experimental results.\label{tbl:results}}
\end{table}

To help understand the role of the different solvers involved in the synthesis process, we provide
a breakdown of how often each solver ``won'', i.e.~was the first to return an answer.
This breakdown is shown in Table~\ref{tbl:wins}.
We see that GP and explicit account for the great majority of the responses, with the load spread
fairly evenly between them.  This distribution illustrates the different strengths of each solver:
GP is very good at generating candidates, explicit is very good at finding counterexamples
and {\sc CBMC} is very good at proving that candidates are correct.  The GP and explicit numbers are
similar because they are approximately ``number of candidates found'' and ``number of  candidates refuted''
respectively.  The {\sc CBMC} column is approximately ``number of candidates proved correct''.
The spread of winners here shows that each of the search strategies is contributing something to the
overall search and that the strategies are able to co-operate with each other.

\begin{table}
\centering
\subfloat[a][How often each solver ``wins''.]{
\begin{tabular}{|c|c|c|c|}
\hline
 {\sc CBMC} & Explicit & GP & Total \\
 \hline
 140 & 510 & 504 & 1183 \\
 12\% & 46\% & 42\% & 100\% \\
 \hline
\end{tabular}
\label{tbl:wins}
}
%\vspace{1em}
%
%\caption{How often each solver ``wins''.\label{tbl:wins}}
\qquad\qquad
\subfloat[b][Where the time is spent.]{
\begin{tabular}{|c|c|c|c|}
\hline
{\sc synth} & {\sc verif} & {\sc generalize} & Total \\
\hline
389.2\,s & 175.8\,s & 25.6\,s & 565.2\,s \\
69\% & 31\% & 5\% & 100\% \\
\hline
\end{tabular}
\label{tbl:time}
}
%
%\vspace{1em}
%\caption{Where the time is spent.\label{tbl:time}}
%
\caption{Statistics about the experimental results.}
\end{table}

To help understand where the time is spent in our solver, Table~\ref{tbl:time} shows how much time is
spent in {\sc synth}, {\sc verif} and constant generalization.  Note that
generalization counts towards {\sc verif}'s time.  We can see that synthesising candidates takes
longer than verifying them, but the ratio of around 2:1 is interesting in that neither phase
completely dominates the other in terms of runtime cost.  This suggests there is great potential
in optimising either of these phases.

\section{Conclusions and Future Work}

We have shown that Second-Order SAT is a very expressive logic occupying
a high complexity class.  %Despite its complexity, it can be reduced to the synthesis
%of finite-state programs, which allows us to exploit the observation that
%many formulae have simple satisfying assignments and that this corresponds to
%synthesising short programs.  
We have also demonstrated that it is
well suited to program verification by directly encoding safety and liveness
properties as Second-Order SAT formulae. Moreover, other applications, such as superoptimisation and QBF solving, map naturally
onto Second-Order SAT.

We built a decision procedure for Second-Order SAT problem via a reduction to 
finite state program synthesis. 
%, we have presented a novel 
The synthesis algorithm is novel and uses a combination of symbolic model checking, explicit
state model checking and stochastic search. 
%Our experiments show that this
%combination is effective at finding short solutions to second-order SAT problems
%stemming from a range of problem domains.
We have observed that for such a difficult problem, surprisingly many instances
can be solved fairly rapidly. This is because the runtime of the solver is strongly influenced by
the size of the shortest solution to the problem, and it seems that many real-world problems have
short solutions.

%% \paragraph{Future Work}

%% There is plenty of scope for continuing work on second-order SAT.  We are interested
%% in applying it to more problem domains including program safety, bug finding and
%% analysis of dynamically allocated data structures.  Additionally, there are a great
%% many ways to improve the performance of our solver.  We are currently investigating
%% domain-specific genetic operators, further parallelising the search and generalising
%% counterexamples.

\paragraph{Future Work.}
%\label{sec:future-work}
%
There is a lot of scope for continuing work on Second-Order SAT.  There are many more problems
that can be encoded as Second-Order SAT, and many opportunities for building more efficient
solvers.
%
%More fundamentally, we have observed that for such a difficult problem, surprisingly many instances
%can be solved fairly rapidly.  This is because the runtime of the solver is so strongly influenced by
%the size of the shortest solution to the problem, and it seems that many real-world problems have
%short solutions.  
One potential application is generalisation. We note that many program analysis problems are concerned with finding properties
of a program that can be generalised, then finding ways of doing such a generalisation.  We conjecture
that \emph{every} successful generalisation strategy would correspond to going from a long proof
to a shorter one, and so consequently that there is a link between the class of tractable program
analysis problems and the class of Second-Order SAT problems with short solutions.

%% \paragraph{Acknowledgements}
%% We thank Martin Brain for his insight and advice during our synthesis discussions.
%% Also, many thanks are due to Cristina David for her extensive help and proofreading during the
%% preparation of this paper.

\bibliography{all,thesis}{}
\bibliographystyle{splncs}

\appendix

\section{Proofs}
\label{app:proofs}

\subsection{Program Encodings}
%Now we turn to the problem of how to encode such finite-state programs.
%For the remainder of this thesis, 
We encode finite-state programs as loop-free
imperative programs consisting of a sequence of instructions, each instruction
consisting of an opcode and a tuple of operands.  The opcode specifies which
operation is to be performed and the operands are the arguments on which the
operation will be performed.  We allow an operand to be one of: a constant literal,
an input to the program, or the result of some previous instruction.  Such a program
has a natural correspondence with a combinational circuit.

A sequence of instructions is certainly a natural encoding of a program,
but we might wonder if it is the \emph{best} encoding.
We can show that for a reasonable set of instruction types (i.e.~valid opcodes), this encoding
is optimal in a sense we will now discuss.
An encoding scheme $E$ takes a function $f$ and assigns it a name $s$.  For a given ensemble
of functions $F$ we are interested in the worst-case behaviour of the encoding $E$, that is we are
interested in the quantitiy $$|E(F)| = \max \{ |E(f)| \mid f \in F \}$$
If for every encoding $E'$, we have that $$|E(F)| = |E'(F)|$$ then we say that $E$ is an \emph{optimal encoding}
for $F$.  Similarly if for every encoding $E'$, we have $$O(|E(F)|) \subseteq O(|E'(F)|)$$ we say that $E$ is
an \emph{asymptotically optimal encoding} for $F$.

\begin{lemma}[Languages with ITE are Universal and Optimal Encodings for Finite Functions]
\label{lem:ite-size}
 For an imperative programming language including instructions
 for testing equality of two values (EQ) and an if-then-else
 (ITE) instruction, any total function $f : \mathcal{S} \to \mathcal{S}$
 can be computed by a program of size $O(| \mathcal{S} | \log | \mathcal{S} |)$ bits.
\end{lemma}

\begin{proof}
The function $f$ is computed by the following program:

 \begin{verbatim}
t1 = EQ(x, 1)
t2 = ITE(t1, f(1), f(0))
t3 = EQ(x, 2)
t4 = ITE(t3, f(2), t2)
...
 \end{verbatim}

Each operand can be encoded in $\log_2 (| \mathcal{S} | + l) = \log_2 (3 \times | \mathcal{S} |)$ bits.
So each instruction can be encoded in $O(\log | \mathcal{S} |)$ bits and there are $O(|\mathcal{S}|)$
instructions in the program, so the whole program can be encoded in $O(| \mathcal{S} | \log | \mathcal{S} |)$
bits.
\end{proof}

\begin{lemma}
\label{lem:large-encoding}
 Any representation that is capable of encoding an arbitrary total function $f : \mathcal{S} \to \mathcal{S}$
 must require at least $O(| \mathcal{S} | \log | \mathcal{S} |)$ bits to encode some functions.
\end{lemma}

\begin{proof}
 There are $| \mathcal{S} |^{| \mathcal{S} |}$ total functions $f : \mathcal{S} \to \mathcal{S}$.
 Therefore by the pigeonhole principle, any encoding that can encode an arbitrary function must use
 at least $\log_2 (| \mathcal{S} |^{| \mathcal{S} |}) = O(| \mathcal{S} | \log_2 | \mathcal{S} |)$
 bits to encode some function.
\end{proof}

From Lemma~\ref{lem:ite-size} and Lemma~\ref{lem:large-encoding}, we can conclude that \emph{any}
set of instruction types that include ITE is an asymptotically optimal function encoding for total functions with
finite domains.

\subsection{Complexity of Finite State Program Synthesis}

\begin{theorem}[Second-Order SAT is Polynomial Time Reducible to Finite Synthesis]
 Every instance of Definition~\ref{def:2sat} is polynomial time reducible to an instance
 of Definition~\ref{def:finite-synth-formula}.
\end{theorem}

\begin{proof}
 We first Skolemise the instance of definition~\ref{def:2sat} to produce an equisatisfiable
 second-order sentence with the first-order part only having universal quantifiers
 (i.e. bring the formula into Skolem normal form).  This process will have introduced
 a function symbol for each first order existentially quantified variable and taken
 linear time.  Now we just existentially quantify over the Skolem functions, which
 again takes linear time and space.
 The resulting formula is an instance of Definition~\ref{def:finite-synth-formula}.
\end{proof}

\subsection{Soundness and Completeness}

\begin{theorem}\label{thm:synth-sound}
Algorithm~\ref{alg:cegis} is sound -- if it terminates with witness $P$, then
$P \models \sigma$.
\end{theorem}

\begin{proof}
 The procedure {\sc synth} terminates only if {\sc synth} returns ``valid''.  In that
 case, $\exists x . \lnot \sigma(P, x)$ is unsatisfiable and so $\forall x . \sigma(P, x)$ holds.
\end{proof}

\begin{theorem}
 \label{thm:synth-semi-complete}
 Algorithm~\ref{alg:cegis} is semi-complete -- if a solution $P \models \sigma$
 exists then Algorithm~\ref{alg:cegis} will find it.
\end{theorem}

\begin{proof}
 If the domain $X$ is finite then the loop in procedure {\sc synth} can only
 iterate $| X |$ times, since by this tProofsime all of the elements of $X$ would have been
 added to the inputs set.  Therefore if the {\sc synth} procedure always terminates,
 Algorithm~\ref{alg:cegis} does as well.

 Since the {\sc ExplicitSearch} routine enumerates all programs (as can be seen by induction on
 the program length $l$), it will eventually enumerate a program that meets the specification
 on whatever set of inputs are currently being tracked, since by assumption such a program
 exists.  Since the first-order theory is
 decidable, the query in {\sc verif} will succeed for this program, causing the algorithm to terminate.
 The set of correct programs is therefore recursively enumerable and Algorithm~\ref{alg:cegis}
 enumerates this set, so it is semi-complete.
\end{proof}

\begin{theorem}
 \label{thm:finite-synth-complete}
 Algorithm~\ref{alg:cegis} with the stopping condition described in Section~\ref{sec:stopping-condition}
 is complete when instantiated with \newC as a background theory -- it will terminate for all specifications $\sigma$.
\end{theorem}

\begin{proof}
 If the specification is satisfiable then Theorem~\ref{thm:synth-semi-complete} holds, and if it is not
 then the stopping condition will eventually hold at which point we (correctly) terminate with an ``unsatisfiable'' verdict.
\end{proof}

\end{document}